\documentclass[%
 reprint,
superscriptaddress,
 amsmath,amssymb,
 aps,
pra,
]{revtex4-1}

\usepackage{amsthm}
\usepackage{graphicx}
\usepackage{dcolumn}
\usepackage{bm}
\usepackage{multirow} 
\usepackage{enumerate}
\usepackage{algorithm}
\usepackage{algorithmic}
\usepackage{tabularx}
\usepackage{hyperref}
\usepackage[mathlines]{lineno}

\newcommand{\abs}[1]{\left\lvert#1\right\rvert}
\newcommand{\norm}[1]{\left\lVert#1\right\rVert}
\newcommand{\ket}[1]{\left|#1 \right\rangle}
\newcommand{\bra}[1]{\left\langle #1 \right|}

\newcommand{\ketbra}[2]{\left| #1 \rangle \langle #2 \right|}

\theoremstyle{remark}

\newtheorem{theorem}{\indent \emph{\textbf{Theorem}}}

\DeclareMathOperator{\polylog}{polylog}

\begin{document}

\makeatletter
\newcommand{\rmnum}[1]{\romannumeral #1}
\newcommand{\Rmnum}[1]{\expandafter\@slowromancap\romannumeral #1@}
\makeatother

\preprint{APS/123-QED}
\title{Quantum algorithm for visual tracking}

\author{Chao-Hua Yu}
\email{quantum.ych@gmail.com}
\affiliation{State Key Laboratory of Networking and Switching Technology, Beijing University of Posts and Telecommunications, Beijing 100876, China}
\affiliation{Department of Physics, The University of Western Australia, Perth 6009, Australia}

\author{Fei Gao}
\affiliation{State Key Laboratory of Networking and Switching Technology, Beijing University of Posts and Telecommunications, Beijing 100876, China}

\author{Chenghuan Liu}
\affiliation{Department of Computer Science and Software Engineering, The University of Western Australia, Perth 6009, Australia}

\author{Du Huynh}
\affiliation{Department of Computer Science and Software Engineering, The University of Western Australia,
Perth 6009, Australia}

\author{Mark Reynolds}
\affiliation{Department of Computer Science and Software Engineering, The University of Western Australia,
Perth 6009, Australia}

\author{Jingbo Wang}
\email{jingbo.wang@uwa.edu.au}
\affiliation{Department of Physics, The University of Western Australia, Perth 6009, Australia}

\date{\today}

\begin{abstract}
Visual tracking (VT) is the process of locating a moving object of interest in a video. It is a fundamental problem in computer vision, with various applications in human-computer interaction, security and surveillance, robot perception, traffic control, etc. In this paper, we address this problem for the first time in the quantum setting, and present a quantum algorithm for VT based on the framework proposed by Henriques et al. [IEEE Trans. Pattern Anal. Mach. Intell., 7, 583 (2015)]. Our algorithm comprises two phases: training and detection. In the training phase, in order to discriminate the object and background, the algorithm trains a ridge regression classifier in the quantum state form where the optimal fitting parameters of ridge regression are encoded in the amplitudes. In the detection phase, the classifier is then employed to generate a quantum state whose amplitudes encode the responses of all the candidate image patches. The algorithm is shown to be polylogarithmic in scaling, when the image data matrices have low condition numbers, and therefore may achieve exponential speedup over the best classical counterpart. However, only quadratic speedup can be achieved when the algorithm is applied to implement the ultimate task of Henriques's framework, i.e., detecting the object position. We also discuss two other important applications related to VT: (1) object disappearance detection and (2) motion behavior matching, where much more significant speedup over the classical methods can be achieved. This work demonstrates the power of quantum computing in solving computer vision problems.

\begin{description}
\item[PACS numbers]
03.67.Dd, 03.67.Hk
\end{description}
\end{abstract}

\pacs{Valid PACS appear here}

\maketitle

\section{Introduction}

Visual tracking (VT) is the task of locating a moving object of interest in a video.
It is a fundamental problem in computer vision and has wide applications 
\cite{HCMB15,HCMB12,SCCDS14,YSZWS11,MSY10} across human-computer interaction, security and surveillance, robot perception, traffic control, and medical imaging. In recent years, a successful approach for VT, which attracts wide attention, has been tracking-by-detection, where discriminative machine learning classifiers are adopted to detect the object. In this approach, every pair of time contiguous frames respectively undergoes two phases: training and detection.  For convenience, hereafter we call the frame used in the training phase the \textit{training frame}, where the location of the object is determined (given in the initial frame or detected in the following frames), and we call the frame used in the detection phase the \textit{detection frame}, where the location of the object remains to be determined. In the training phase, by using discriminative machine learning algorithms, a number of image patches (samples) around the object are selected from the training frame to train a classifier that can discriminate between the object and its background. In the detection phase, the classifier is then performed on 
several candidate patches selected from the detection frame to calculate their responses.  The maximum response  reveals the most probable
position of the object. The training-detection procedure is run successively over every frame to track the object in the whole video, during which the detection frame becomes the training frame once the object is detected.  The whole process is depicted in Fig.~\ref{fig:TBDVT}. Despite 
various advanced and fast algorithms \cite{HCMB15,HCMB12}, VT can be time consuming when the processed data size is large.

\begin{figure}[b]
\includegraphics[scale=0.382]{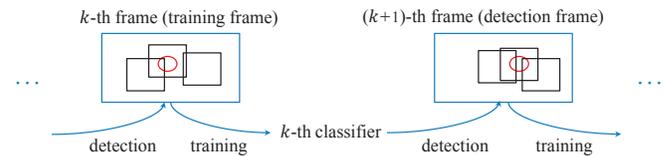}
\caption{\label{fig:TBDVT} Schematic of visual tracking by training and detection.  
The (red) circles represent the object in the training  and detection frame.  The (black) rectangles in the training frame denote sample image patches for training, and those in the detection frame denote the candidate image patches for detection. The training-detection procedure is run successively over contiguous frames to track the object in the video.}
\end{figure}

Quantum computing provides a paradigm that makes use of quantum mechanical principles, such as superposition and entanglement, to perform computing tasks in quantum systems (quantum computers) \cite{QCQI10,QA15}. The most exciting thing about quantum computing is that it can achieve significant speedup over  classical computing, in solving certain problems, such as simulating quantum systems \cite{Qsim17}, factoring large numbers \cite{Shor94}, unstructured database searching \cite{Grover97,AA02} and solving linear systems of equations by Harrow et al. (HHL algorithm) \cite{HHL09}.
In recent years, the applicability 
of quantum computing
has been 
extended to the fields of machine learning and data mining, resulting in a variety of related quantum algorithms \cite{WBL12,SSP16,Wang17,QRR17,QSVM14,YGWW16,LR18} for solving various machine learning and data mining problems, such as linear regression \cite{WBL12,SSP16,Wang17}, support vector machine \cite{QSVM14}, and association rules mining \cite{YGWW16}. These algorithms are shown to achieve significant speedup over their classical counterparts. Overviews on the recent progress in this field can be found in the references \cite{DB17,QML17}.

Motivated by the progress in VT as well as quantum computing,
and especially in machine learning, we explore whether and how quantum computing can be exploited to implement VT
more efficiently than classical computing. In particular, we propose a quantum algorithm for VT based on the well-known tracking-by-detection VT framework proposed by Henriques et al. in 2015 (HCMB15 framework) \cite{HCMB15}. The earlier version of this framework is given in the reference \cite{HCMB12}. In the training phase, the HCMB15 framework utilizes a \emph{base sample patch} to subtly produce a large number of \emph{virtual sample patches} which can be represented by a circulant data matrix, and use these samples to train a ridge regression classifier. In the detection phase, it uses a \textit{base candidate patch} to generate a large number of \emph{virtual candidate patches} which can also be represented by a circulant matrix, performs the classifier on those candidate patches, and obtains their corresponding responses. The most probable location of the object is revealed by the candidate patch with the maximum response. The circulant structure of
the 
data matrix has been cleverly exploited in the HCMB15 framework
and makes it extremely efficient in both training and detection phases.

Our quantum algorithm also comprises the training phase and the detection phase. In the training phase, a ridge regression classifier is trained in the quantum state form, where the optimal fitting parameters of ridge regression are encoded in the amplitudes. In the detection phase, the classifier is then utilized to generate a quantum state whose amplitudes encode the responses of all the candidate image patches. The whole algorithm is built on the proposed subroutine of simulating extended circulant Hamiltonians, which allows the estimation of the singular values of generic circulant matrices by phase estimation. Our algorithm generates both states in polylogarithmic time when the data matrices have low condition numbers, demonstrating its potential to achieve exponential speedup over the classical counterpart. Moreover, we also show that the output quantum state of our algorithm can be fully used to efficiently implement two important computer vision tasks, object disappearance detection and motion behavior matching.

The rest of the paper is organized as follows. In Sec.~\ref{Sec:HCMB15},
we review the HCMB15 framework in terms of its basic concepts, notations,
and classical algorithmic procedures. Sec.~\ref{Sec:QVT} presents and analyzes our quantum algorithm. Sec.~\ref{Sec:Applications} gives two important applications of our algorithm. Finally, conclusions are
drawn in the last section.

\section{HCMB15 framework}
\label{Sec:HCMB15}
In this section, we review the basic idea and algorithms in the HCMB15 framework \cite{HCMB15,HCMB12}. For simplicity, we only consider one-dimensional images with single channel. The generalization to two-dimensional multiple-channel images can be seen in \cite{HCMB15,HCMB12}.

The HCMB15 framework is one of the excellent candidates for implementing VT. In the training phase, it takes a base sample image patch of the training frame (where the object is usually placed at the center) with $n$ pixels that can be represented by a vector $\mathbf{x}=(x_1,x_2,\cdots,x_n)^T$, the size of which is, for example, two times that of the object. By cyclic shifting $\mathbf{x}$, it can be used to generate $n$ virtual samples corresponding to a circulant matrix
\begin{eqnarray}
\label{Eq:CirMatrX}
X=\mathcal{C}(\mathbf{x})=\left[
\begin{array}{ccccc}
x_1     &x_2 &x_3 &\cdots &x_n \\
x_{n}   &x_1 &x_2 &\cdots &x_{n-1}\\
\vdots  &\vdots &\vdots &\ddots &\vdots \\
x_2     &x_3 &x_4 &\cdots &x_1 \\
\end{array}
\right],
\end{eqnarray}
where $\mathcal{C}(\mathbf{x}):\mathbb{R}^n\rightarrow \mathbb{R}^{n\times n}$ is a function which generates the circulant matrix for a given vector $\mathbf{x}$. The $i$-th row of $X$ is denoted by $\mathbf{x}_i$ with $\mathbf{x}_1=\mathbf{x}$, and corresponds to the $i$-th training sample.

In addition, each sample is assigned a \textit{label} (or \textit{regression target}), a positive value ranging from 0 to 1, to quantify the closeness between the sample and the base sample; the value approximates to 1 if the sample is close to the base sample, and reduces to 0 as the distance between them increases. The label of $\mathbf{x}_i$, denoted by $y_i$, is commonly derived using the Gaussian function
\begin{eqnarray}
\label{Eq:TrainOutput}
y_i=e^{-d_i^2/s^2},
\end{eqnarray}
where $d_i$ is the Euclidean distance of the $i$-th sample to the base sample in the image, and $s$ is the bandwidth and is commonly taken as $s=c\sqrt{n}$ for some constant $c$.


In the training phase, the goal is to train a linear function $f(\mathbf{x})=\mathbf{w}^T\mathbf{x}$ by \textit{ridge regression}, which minimizes the squared error over samples $\mathbf{x}_i$ and $y_i$,
\begin{eqnarray}
\label{Eq:TrainOpt}
\min\limits_{\mathbf{w}} \sum_{i=1}^n \abs{f(\mathbf{x}_i)-y_i}^2+\alpha\norm{\mathbf{w}}^2,
\end{eqnarray}
where $\alpha$ is the regularization parameter that controls overfitting. The solution is given by
\begin{eqnarray}
\label{Eq:TrainSolution}
\mathbf{w}=(X^TX+\alpha I)^{-1}X^T\mathbf{y},
\end{eqnarray}
where $\mathbf{y}=(y_1,y_2,\cdots,y_n)^T$ is the vector of the regression targets of all $n$ samples, and $X$ is called the \textit{data matrix} \cite{HCMB15}. After the solution $\mathbf{w}$ is obtained, one can predict the response of a new image patch $\hat{\mathbf{\mathbf{x}}}$ by calculating $f(\hat{\mathbf{\mathbf{x}}})=\mathbf{w}^T\hat{\mathbf{\mathbf{x}}}$.

The HCMB15 framework takes advantage of the property of the circulant matrix $X$ that it can be written as $X=F \textmd{diag}(\mathcal{F}(\mathbf{x}))F^{\dag}$. Here $F$ is the unitary Fourier transformation matrix, $\mathcal{\mathcal{F}}(\mathbf{x})=\sqrt{n}F\mathbf{x}$ is the Discrete Fourier Transformation (DFT), and $\textmd{diag}(\mathbf{z})$ is the diagonal matrix formed by $\mathbf{z}$. As a result, the solution (Eq.~\eqref{Eq:TrainSolution})) can be efficiently obtained with time complexity $O(n\log(n))$ \cite{HCMB15}, which is significantly faster than the
currently
prevalent method by matrix inversion and products that has time complexity $O(n^3)$.

In the detection phase, a base candidate image patch of the detection frame denoted by a $n$-dimensional vector $\mathbf{z}$ is given, and is used to generate $n$ virtual candidate patches corresponding to the $n\times n$ matrix $Z=\mathcal{C}(\mathbf{z})$ with the $i$-th row corresponding to the $i$-th candidate patch. Then the responses of these patches are predicted by
\begin{eqnarray}
\label{Eq:DetectionPred}
\hat{\mathbf{y}}= Z\mathbf{w},
\end{eqnarray}
where the $i$-th element of $\hat{\mathbf{y}}$ corresponds to the response of $i$-th candidate patch. The index with maximum response in $\hat{\mathbf{y}}$ reveals the target image patch that gives the best estimated position of the object and serves as the new base sample for the next training-detection procedure.

\section{Quantum algorithm}
\label{Sec:QVT}
In this section, we present a quantum algorithm for VT based on the HCMB15 framework with focus on one-dimensional single-channel images. Just as the classical HCMB15 framework, our algorithm also comprises two phases: training and detection. In the first place, we propose a technique of simulating the extended circulant Hamiltonians in Sec.~\ref{Subsec:QVT:EXCHS}. By taking this technique as the basic subroutine, we construct our quantum algorithm to prepare the quantum state $\ket{\mathbf{w}}$ (normalized $\mathbf{w}$) in the training phase shown in Sec.~\ref{Subsec:QVT:Training}, and then ultimately generate the state $\ket{\hat{\mathbf{y}}}$ (normalized $\hat{\mathbf{y}}$) in the detection phase shown in Sec.~\ref{Subsec:QVT:Detection}. We perform runtime analysis on our algorithm in Sec.~\ref{Subsec:QVT:Runtime}, and finally extend it to two-dimensional single-channel images in Sec.~\ref{Subsec:QVT:Extension}.

\subsection{Extended circulant Hamiltonian simulation}
\label{Subsec:QVT:EXCHS}

Since Eq.~\eqref{Eq:CirMatrX} describes a general $n \times n$ circulant matrix, we let $X$ denote an arbitrary $n \times n$ circulant matrix. In the case where $X$ is Hermitian, i.e., $X=X^{\dag}$, Zhou and Wang \cite{ZW17} proposed an efficient quantum algorithm that uses the unitary linear decomposition approach \cite{BCCKS15} to implement $e^{-iXt}$ within spectral-norm error $\epsilon$ in time $O(t\polylog(n)\log(t/\epsilon)/\log\log(t/\epsilon))$, under the assumptions that the quantum oracle $O_\mathbf{x}\ket{0}^{\otimes \lceil \log n \rceil}=\sum_{i=1}^n \sqrt{x_i}\ket{i}$ that can be efficiently implemented in time $O(\polylog(n))$ is provided, and $\sum_{i=1}^n  x_i=1$. This algorithm is based on the observation that $X$ can be written as a linear combination of $n$ efficient-to-implement unitary operators, namely
\begin{eqnarray}
X=\sum_{j=1}^n x_jV_{j},
\end{eqnarray}
where each $V_j=\sum_{l=0}^{n-1} \ket{(l-j+1)\mod n}\bra{l}$ for $j=1,2,\cdots,n$, and can be implemented using $O(\log n)$ one- or two-qubit gates.

However, $X$ in general is not Hermitian and can not be simulated directly by the algorithm \cite{ZW17}. To overcome this, we take the \textit{extended circulant Hamiltonian} of $X$,
\begin{eqnarray}
\label{Eq:ExtMatrX}
\tilde{X}&=&\ket{0}\bra{1}\otimes X+\ket{1}\bra{0}\otimes X^{\dag} \nonumber\\
&=& \left[
\begin{array}{cc}
0 & X \\
X^{\dag} &0\\
\end{array}
\right].
\end{eqnarray}
It is interesting and easy to see that $\tilde{X}$ can also be written as a linear combination of simple unitary operators:
\begin{eqnarray}
\label{Eq:LinCombOfXtilde}
\tilde{X}&=&\sum_{j=1}^{n}x_j(\ket{0}\bra{1}\otimes V_j+\ket{1}\bra{0}\otimes V_j^{\dag}) \nonumber\\
&=& \sum_{j=1}^{n}x_j(\sigma_X\otimes I)(\ket{0}\bra{0}\otimes V_j+\ket{1}\bra{1}\otimes V_j^{\dag}), \nonumber\\
&=& \sum_{j=1}^{n}x_j\tilde{V}_j,
\end{eqnarray}
where $\sigma_X$ is the Pauli-X gate (or NOT gate).
Therefore, following Zhou and Wang \cite{ZW17}, we can also use the unitary linear decomposition approach \cite{BCCKS15} to design an efficient quantum algorithm that implements $e^{-i\tilde{X}t}$ within some (spectral-norm) error $\epsilon$. We present the result in the following theorem.

\begin{theorem}
\label{Theorem:ExCirSimulation}
(\textbf{Extended circulant Hamiltonian simulation}) There exists a quantum algorithm that implements $e^{-i\tilde{X}t}$ within error $\epsilon$ by taking $O\left(t\log(t/\epsilon)/\log\log(t/\epsilon)\right)$ calls of controlled-$O_{\mathbf{x}}$ ($\ketbra{0}{0}\otimes I +\ketbra{1}{1}\otimes O_{\mathbf{x}}$), and $O\left(t\log(n)\log(t/\epsilon)/\log\log(t/\epsilon)\right)$ one- or two-qubit gates.
\end{theorem}

\begin{proof}
This theorem can be readily proved by following the same argument as proving theorem 4.1 in \cite{ZW17}. In that proof, a quantum algorithm which involves a series of controlled-$O_{\mathbf{x}}$ together with controlled $V_j$ and their inverse was constructed. To prove the above theorem, we can construct a similar quantum algorithm where $\tilde{V}_j$ instead of $V_j$ are used. It is easy to see from Eq.~\eqref{Eq:LinCombOfXtilde} that $\tilde{V}_j$ can also be efficiently implemented by taking $O(\log n)$ one- or two-qubit gates just as $V_j$.
\end{proof}

According to this theorem,  $e^{-i\tilde{X}t}$ can be efficiently implemented within error $\epsilon$ with time complexity $O\left(t\polylog(n)\log(t/\epsilon)/\log\log(t/\epsilon)\right)$, under the assumption that $O_{\mathbf{x}}$ can be implemented in time $O(\polylog(n))$. This efficient implementation of $e^{-i\tilde{X}t}$ for an arbitrary circulant matrix $X$ will serve as an elementary subroutine in both phases of our algorithm as shown in the following.

\subsection{Algorithm phase \Rmnum{1}: Training}
\label{Subsec:QVT:Training}

The training phase aims to produce the quantum state $|\mathbf{w}\rangle$. To attain a more concise form of $\mathbf{w}$ (as well as $|\mathbf{w}\rangle$), we write $X$ in the singular value decomposition form
\begin{eqnarray*}
X=\sum_{j=1}^n\lambda_j\ket{\mathbf{u}_j}\bra{\mathbf{v}_j},
\end{eqnarray*}
where $\{\lambda_j\}_{j=1}^n$, $\{\ket{\mathbf{u}_j}\}_{j=1}^n$, and $\{\ket{\mathbf{v}_j}\}_{j=1}^n$ are respectively the singular values, the left singular vectors, and the right singular vectors of $X$. So, according to the definition of Eq.~\eqref{Eq:ExtMatrX}, $\tilde{X}$ has $2n$ eigenvalues $\{\pm \lambda_j\}_{j=1}^{n}$ and eigenvectors $\{\ket{\mathbf{w}_j^{\pm}}\}_{j=1}^{n}$, where
\begin{eqnarray*}
\ket{\mathbf{w}_j^{\pm}}=(\ket{0}\ket{\mathbf{u}_j}\pm\ket{1}\ket{\mathbf{v}_j})/\sqrt{2}.
 \end{eqnarray*}
Since $\{\ket{\mathbf{u}_j}\}_{j=1}^n$ constitutes an orthonormal basis of the $\mathbb{R}^n$ space, $\mathbf{y}$ can be written as a linear combination of these basis vectors, i.e., $\mathbf{y}=\sum_{j=1}^n\beta_j\norm{\mathbf{y}}\ket{\mathbf{u}_j}$. In this case, the ridge regression solution $\mathbf{w}$ (Eq.~\eqref{Eq:TrainSolution}) can be rewritten as
\begin{eqnarray}
\label{Eq:TrainSolutionSVD}
\mathbf{w}=\sum_{j=1}^n\frac{\beta_j\lambda_j\norm{\mathbf{y}}}{\lambda_j^2+\alpha}\ket{\mathbf{v}_j}.
\end{eqnarray}
It is notable that a good $\alpha$, with which ridge regression can achieve good predictive performance, can be chosen efficiently by quantum cross validation \cite{QRR17}.


The training phase of our algorithm to generate $\ket{\mathbf{w}}$ (normalized Eq.~\eqref{Eq:TrainSolutionSVD}) is detailed as below, and the corresponding quantum circuit is shown in Fig.~\ref{fig:QCircuit1}.

\begin{figure}[htb]
\includegraphics[scale=0.4]{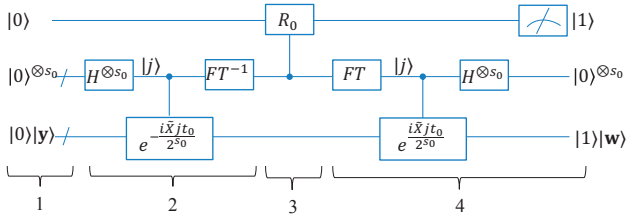}
\caption{\label{fig:QCircuit1} Quantum circuit for the training phase. Here the numbers 1, 2, 3 and 4 denote the sequence of four steps of the training phase, '/' denotes a bundle of wires, $H$ denotes the Hadamard operation, $FT$ represents the quantum Fourier transformation \cite{QCQI10}, and controlled-$R_0$ denotes the controlled rotation in the training step 3.}
\end{figure}

1. Prepare three quantum registers in the state $(\ket{0}\ket{\mathbf{y}})(\ket{0}^{\otimes s_0})\ket{0}$, where
\begin{eqnarray}
\ket{0}\ket{\mathbf{y}}\nonumber
=&&\ket{0}(\sum_{j=1}^n\beta_j\ket{\mathbf{u}_j}) \nonumber\\
=&& \sum_{j=1}^n\beta_j\left(\frac{\ket{\mathbf{w}_j^{+}}+\ket{\mathbf{w}_j^{-}}}{\sqrt{2}}\right),
\end{eqnarray}
where $s_0$ denotes the number of qubits used for storing eigenvalues in the next step.
Consequently, the regression targets $(y_1,y_2,\cdots,y_n)$ of all the samples are encoded in the amplitudes of $\ket{\mathbf{y}}$. The details of preparing $\ket{\mathbf{y}}$ are shown in the Appendix. 

2. Performing phase estimation of $e^{-i\tilde{X}t_0}$ on the first two registers, we obtain the whole state
\begin{eqnarray}
\sum_{j=1}^n \beta_j\left(\frac{\ket{\mathbf{w}_j^{+}}\ket{\lambda_j}+\ket{\mathbf{w}_j^{-}}\ket{-\lambda_j}}{\sqrt{2}}\right)\ket{0}.
\end{eqnarray}

3. Performing a controlled rotation on the last register (qubit) conditioned on the eigenvalue register, we have
\begin{eqnarray}
&&\sum_{j=1}^n \beta_j\Bigg(\frac{\ket{\mathbf{w}_j^{+}}\ket{\lambda_j}\left(\frac{C\lambda_j}{\lambda_j^2+\alpha}\ket{1}+\sqrt{1-(\frac{C\lambda_j}{\lambda_j^2+\alpha})^2}\ket{0}\right)}{\sqrt{2}} \nonumber\\
&&+\frac{\ket{\mathbf{w}_j^{-}}\ket{-\lambda_j}\left(\frac{-C\lambda_j}{\lambda_j^2+\alpha}\ket{1}+\sqrt{1-(-\frac{C\lambda_j}{\lambda_j^2+\alpha})^2}\ket{0}\right)}{\sqrt{2}} \Bigg).
\end{eqnarray}
Here $C=O(\min_{j}\lambda_j)$.

In other words, the last qubit is rotated by the angle
$$\theta(\lambda)=\arcsin\left(\frac{C\lambda}{\lambda^2+\alpha}\right)$$
conditioned on $\ket{\lambda}$, where $\lambda$ is the eigenvalue. To implement this, we use the methodology presented in \cite{CD16} to construct the register $\ket{\theta(\lambda)}$ from $\ket{\lambda}$, where $\theta(\lambda)$ is approximated by truncating the Taylor series of $\theta(\lambda)$ (of $\lambda$) to some order, and $\ket{\theta(\lambda)}$ is accordingly obtained by performing a sequence of quantum multiplication and quantum addition operations as described in \cite{CD16}.

4. Undoing phase estimation and measuring the last qubit to see the outcome $\ket{1}$, we obtain the state of the first register
\begin{eqnarray}
&&\ket{1}\left(\sum_{j=1}^n\frac{C\beta_j\lambda_j}{\lambda_j^2+\alpha}\ket{\mathbf{v}_j}/\sqrt{\sum_{j=1}^n \left(\frac{C\beta_j\lambda_j}{\lambda_j^2+\alpha}\right)^2}\right) \nonumber\\
=&&\ket{1}\ket{\mathbf{w}}.
\end{eqnarray}
Discarding $\ket{1}$, we derive the quantum state $\ket{\mathbf{w}}$ as desired.

\subsection{Algorithm phase \Rmnum{2}: Detection}
\label{Subsec:QVT:Detection}

Obtaining $\ket{\mathbf{w}}$ via the training phase of our algorithm allows us to proceed to the next phase: detection. In the detection phase, according to Eq.~\eqref{Eq:DetectionPred}, we are to produce the final state $\ket{\hat{\mathbf{y}}}$ by performing the operation $Z$ on $\ket{\mathbf{w}}$. To see $\ket{\hat{\mathbf{y}}}$ more concisely, let $Z$ be written in the singular value decomposition form, $Z=\sum_{j=1}^n\gamma_j\ket{\mathbf{p}_j}\bra{\mathbf{q}_j}$, where $\{\gamma_j\}_{j=1}^n$, $\{\ket{\mathbf{p}_j}\}_{j=1}^n$ and $\{\ket{\mathbf{q}_j}\}_{j=1}^n$ are respectively the singular values, the left singular vectors, and the right singular vectors of $Z$. Since $\ket{\mathbf{w}}$ lies in the space spanned by $\{\ket{\mathbf{q}_j}\}_{j=1}^n$, $\ket{\mathbf{w}}$ can be written as a linear combination of them, i.e., $\ket{\mathbf{w}}=\sum_{j=1}^n\delta_j \ket{\mathbf{q}_j}$. So $\ket{\hat{\mathbf{y}}}$ can be rewritten as
\begin{eqnarray}
\ket{\hat{\mathbf{y}}}=\sum_{j=1}^n\gamma_j\delta_j \ket{\mathbf{p}_j}/\sqrt{\sum_{j=1}^n\gamma_j^2\delta_j^2}.
\end{eqnarray}

Similar to the training phase, this phase resorts to the simulation of the extended circulant Hamiltonian of $Z$, $\tilde{Z}=\ket{0}\bra{1}\otimes Z+\ket{1}\bra{0}\otimes Z^{\dag}$, which has $2n$ eigenvalues $\{\pm \gamma_j\}_{j=1}^{n}$ and eigenvectors $\{\ket{\mathbf{r}_j^{\pm}}=(\ket{0}\ket{\mathbf{p}_j}\pm\ket{1}\ket{\mathbf{q}_j})/\sqrt{2}\}_{j=1}^{n}$. Here we also assume that a quantum oracle $O_{\mathbf{z}}$ that efficiently implements $\ket{0}^{\log \lceil n \rceil}\mapsto \sum_{i=1}^n \sqrt{z_i}\ket{i}$ in time $O(\polylog(n))$ and are provided in the controlled fashion, and that $\sum_{i=1}^n z_i=1$. Under these assumptions, we are able to efficiently implement $e^{-i\tilde{Z}t}$ by Theorem 1. Armed with this ability, the detection phase of our algorithm to generate $\ket{\hat{\mathbf{y}}}$ can be achieved by the following steps, and the corresponding quantum circuit is shown in Fig.~\ref{fig:QCircuit2}.

1. Prepare three quantum registers in the state $(\ket{1}\ket{\mathbf{w}})(\ket{0}^{\otimes s_1})\ket{0}$, where
\begin{eqnarray}
\ket{1}\ket{\mathbf{w}}
&=&\ket{1}(\sum_{j=1}^n\delta_j\ket{\mathbf{q}_j}) \nonumber\\
&=& \sum_{j=1}^n\delta_j\left(\frac{\ket{\mathbf{r}_j^{+}}-\ket{\mathbf{r}_j^{-}}}{\sqrt{2}}\right),
\end{eqnarray}
and $s_1$ denotes the number of qubits for phase estimation in the following step. Note that $\ket{\mathbf{w}}$ has been produced as shown in the training phase.

2. Performing phase estimation of $e^{-i\tilde{Z}t_1}$ on the first two registers, we have
\begin{eqnarray}
\sum_{j=1}^n\delta_j\left(\frac{\ket{\mathbf{r}_j^{+}}\ket{\gamma_j}-\ket{\mathbf{r}_j^{-}}\ket{-\gamma_j}}{\sqrt{2}}\right)\ket{0}.
\end{eqnarray}

3. Similar to the third step of the training phase, we also perform a controlled rotation operation on the last two registers to obtain the whole state
\begin{eqnarray}
&&\sum_{j=1}^n \delta_j\bigg(\frac{\ket{\mathbf{r}_j^{+}}\ket{\gamma_j}(C'\gamma_j\ket{1}+\sqrt{1-(C'\gamma_j)^2}\ket{0})}{\sqrt{2}} \nonumber\\
&&-\frac{\ket{\mathbf{r}_j^{-}}\ket{-\gamma_j}(-C'\gamma_j\ket{1}+\sqrt{1-(-C'\gamma_j)^2}\ket{0})}{\sqrt{2}} \bigg).
\end{eqnarray}
Here $C'=O(\max_{j}\gamma_j)^{-1}$.

The controlled rotation can be implemented by the same approach presented in step 3 of the training phase, but the function of angle $\theta(\lambda)$ is replaced by
$\theta(\gamma)=\arcsin(C'\gamma)$, where $\gamma$ represents the eigenvalue stored in the second (eigenvalue) register.

4. Undo the phase estimation and measure the last qubit to see the outcome $\ket{1}$. If successful, we obtain the state of the first register,
\begin{eqnarray}
&&\ket{0}\left(\sum_{j=1}^nC'\delta_j\gamma_j\ket{\mathbf{p}_j}/\sqrt{\sum_{j=1}^n \left(C'\delta_j\gamma_j\right)^2}\right) \nonumber\\
=&&\ket{0}\ket{\hat{\mathbf{y}}}.
\end{eqnarray}
Then, discarding $\ket{0}$, we derive the quantum state $\ket{\hat{\mathbf{y}}}$ as desired.

\begin{figure}[htb]
\includegraphics[scale=0.4]{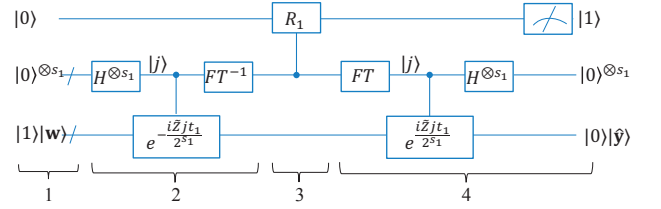}
\caption{\label{fig:QCircuit2} Quantum circuit for generating $\ket{\hat{\mathbf{y}}}$. Here the numbers 1, 2, 3 and 4 denote the four sequential steps of the detection phase, and controlled-$R_1$ denotes the controlled rotation in the detection step 3.}
\end{figure}

\subsection{Runtime analysis}
\label{Subsec:QVT:Runtime}
We respectively analyze the runtime in the training phase and that in the detection phase of our algorithm, and discuss the overall runtime.

In the training phase, the error occurs in the Hamiltonian simulation of $\tilde{X}$ and in the phase estimation in step 2. Since the complexity of Hamiltonian simulation scales sublogarithmically in the inverse of error, as shown in Theorem \ref{Theorem:ExCirSimulation}, and that of the phase estimation scales linearly \cite{QCQI10}, the source of error dominantly comes from the phase estimation. The phase estimation induces error $O(1/t_0)$ in estimating $\lambda$ (singular value of $X$), and relative error
$$O\left(\frac{\lambda^2-\alpha}{t_0\lambda(\lambda ^2+\alpha)}\right)=O(1/t_0\lambda)$$
in estimating $\lambda/(\lambda^2+\alpha)$. Since $\sum_{i=1}^n x_i=1$, the spectral norm of $X$ is 1 and thus $1/\kappa_X \leq \lambda \leq 1$, where $\kappa_X$ is the condition number of $X$. So $t_0=O(\kappa_X/\epsilon)$ induces final error $\epsilon$ according to the analysis in HHL algorithm \cite{HHL09}.
In step 3, $\theta(\lambda)$ can be approximated within error $\epsilon$ by truncating its Taylor series to the order $O\left(\frac{\log(1/\epsilon)}{\log\log(1/\epsilon)}\right)$, meaning that $O\left(\frac{\log(1/\epsilon)}{\log\log(1/\epsilon)}\right)$ iterations of performing quantum multiplication and quantum addition operations are required \cite{CD16}. Since the eigenvalue register $\ket{\lambda_j}$ (and other ancilla registers storing intermediate results) should be of $O\left(\log(\kappa_X/\epsilon)\right)$ qubits to ensure $\lambda_j$ being estimated within error $O(\epsilon/\kappa_X)$ via phase estimation (in step 2)\cite{QCQI10}, each iteration takes $O\left(\log^2(\kappa_X/\epsilon)\right)$ and $O\left(\log(\kappa_X/\epsilon)\right)$ elementary gates to implement quantum multiplication and addition, respectively. So step 3 totally takes time $O\left(\frac{\log(1/\epsilon)\log^2(\kappa_X/\epsilon)}{\log\log(1/\epsilon)}\right)$, which is relatively negligible compared to the time taken in step 2.
In step 4, the success probability of obtaining
$\ket{1}$ is $$\sum_j \left(\frac{C\beta_j\lambda_j}{\lambda_j^2+\alpha}\right)^2=\Omega(1/\kappa_X^2)$$
since $C\lambda_j/(\lambda_j^2+\alpha)=\Omega(1/\kappa_X)$, which means $O(\kappa_X^2)$ measurements are required to obtain $\ket{1}$ with a high probability and this can be improved to $O(\kappa_X)$ repetitions by amplitude amplification. Here $\alpha$ is taken for convenience to be in the range $[1/\kappa_X^2,1]$, because from Eq.~\eqref{Eq:TrainSolution} it is easy to see that when $\alpha$ is too small, ridge regression is reduced to ordinary linear regression, and when $\alpha$ is too large, it is the operation $X$ that is performed in Eq.~\eqref{Eq:TrainSolution}. Therefore, the total time complexity for generating $\ket{\mathbf{w}}$ in the training phase is $T_\mathbf{w}=\tilde{O}(\polylog(n) \kappa_X^2/\epsilon)$, where $\tilde{O}$ is used to suppress the lower growing terms, the polylogarithmic factors in simulating Hamiltonian as shown in Theorem \ref{Theorem:ExCirSimulation}, and in step 3.

In the detection phase, the source of error for generating $\ket{\hat{\mathbf{y}}}$ is also dominated by the phase estimation in step 2. The phase estimation induces relative error $O(1/\gamma)$ in estimating $\gamma$, where $\gamma$ denotes the singular value of $Z$. Since $\sum_{i} z_i=1$, the spectral norm of $Z$ is also one and $1/\kappa_Z \leq \gamma \leq 1$, where $\kappa_Z$ is the condition number of $Z$. Thus $t_1=O(\kappa_Z/\epsilon)$ induces final error $\epsilon$ for generating $\ket{\hat{\mathbf{y}}}$. Just as the step 3 of the training phase, the step 3 of the detection phase also takes relatively negligible time. In step 4, the success probability of obtaining the outcome $\ket{1}$ is $\sum_{j=1}^n \left(C'\delta_j\gamma_j\right)^2=\Omega(1/\kappa_Z^2)$, which means $O(\kappa_Z)$ repetitions are required to obtain $\ket{1}$, with a high probability, by amplitude amplification. Taking into account the time for generating $\ket{\mathbf{w}}$ (in the training phase) in step 1, the runtime of the detection phase of our algorithm scales as $\tilde{O}\left(\kappa_Z(\polylog(n)\kappa_Z/\epsilon+T_\mathbf{w})\right)=\tilde{O}\left(\polylog(n)\kappa_Z(\kappa_Z+\kappa_X^2)/\epsilon\right)$.

The overall time complexity of our algorithm has the same scaling as that of the detection phase, because the training phase to create $\ket{\mathbf{w}}$ has been incorporated into the detection phase in the sense that $\ket{\mathbf{w}}$ is taken as the input in the detection phase (as shown in its first step).
This means that, compared to the classical HCMB15 framework which takes same $O(n\log(n))$ time to obtain both $\mathbf{w}$ and $\hat{\mathbf{y}}$, our algorithm takes only exponentially less time for generating their quantum-state versions $\ket{\mathbf{w}}$ and $\ket{\hat{\mathbf{y}}}$, when $\kappa_X, \kappa_Z, 1/\epsilon=O(\polylog(n))$. Comparisons between the classical HCMB15 framework and our quantum algorithm are detailed in Table \ref{tab:AlgorithmComparison}.

\onecolumngrid

\begin{table}[htb]
\caption{\label{tab:AlgorithmComparison}%
Comparisons between our quantum algorithm and the classical HCMB15 algorithm, in terms of their inputs/outputs and their time complexity.}
\resizebox{\textwidth}{!}{
\begin{ruledtabular}
\begin{tabular}{cccccc}
Algorithms &Inputs &Outputs & Time complexity \\
\hline
 Classical algorithm &$\mathbf{x}$,$\mathbf{\mathbf{z}}$, $\mathbf{y}$, $\alpha$ &$\hat{\mathbf{y}}$ &$O(n\log(n))$ \\
 Quantum   algorithm &Controlled-$O_\mathbf{x}$, Controlled-$O_\mathbf{\mathbf{z}}$, $\ket{\mathbf{y}}$, $\alpha$ &$\ket{\hat{\mathbf{y}}}$ &$\tilde{O}\left(\polylog(n)\kappa_Z(\kappa_Z+\kappa_X^2)/\epsilon\right)$ \\
\end{tabular}
\end{ruledtabular}}
\end{table}
\twocolumngrid

Further runtime needs to be considered if we use the output $\ket{\hat{\mathbf{y}}}$ to implement other specific tasks. For example, let us consider the ultimate task of the HCMB15 framework to detect target candidate patch, which entails identifying the index of largest squared amplitude (denoted by $p_{max}$) in $\ket{\hat{\mathbf{y}}}$ by sampling the state. Then we require $O(1/p_{max})$ samples to reveal $p_{max}$, which will be efficient if $1/p_m=O(\polylog(n))$. However, in practice, when the object is clearly visible in the detection frame, the set of elements in $\hat{\mathbf{y}}$ will be approximately equal to that in $\mathbf{y}$ \cite{PC}. In this case the largest squared amplitude in $\ket{\hat{\mathbf{y}}}$ is close to that in $\ket{\mathbf{y}}$, that is, $p_{max}=1/(\sum_{i=1}^n y_i^2)\approx \frac{2\sqrt{2}}{\sqrt{\pi}s}=O(1/\sqrt{n})$ according to the result of the Appendix. This means that $\Omega(\sqrt{n})$ copies of $\ket{\hat{\mathbf{y}}}$ are required to reveal the largest squared amplitude in $\ket{\hat{\mathbf{y}}}$.
As a result, the overall runtime for this task would be $\Omega\left(\sqrt{n}\polylog(n)\kappa_Z(\kappa_Z+\kappa_X^2)/\epsilon\right)$ by our quantum algorithm, achieving quadratic improvement at most over the classical HCMB15 framework. Hereafter we focus on using $\ket{\hat{\mathbf{y}}}$ to implement other interesting tasks related to VT, which can achieve much more significant speedup over their classical methods, as shown in the following Sec. \ref{Sec:Applications}, rather than sampling it exhaustively.

\subsection{Extension to two-dimensional images}
\label{Subsec:QVT:Extension}

For two-dimensional images, a base sample patch of size $n\times m$ is represented by a $n\times m$ matrix $\mathbf{x}$ with the $j$th ($1 \leq j\leq m$) row denoted by a $n$-dimensional vector $\mathbf{x}_j$. All the samples correspond to cyclic shifts of the base sample in both horizontal and vertical directions, and can be described by a block circulant matrix with circulant blocks \cite{HCMB12}, resulting in the $nm\times nm$ data matrix
\begin{eqnarray}
\label{Eq:BloCirMatrX}
X=\left[
\begin{array}{ccccc}
\mathcal{C}({\mathbf{x}_1})     &\mathcal{C}({\mathbf{x}_2}) &\mathcal{C}({\mathbf{x}_3}) &\cdots &\mathcal{C}({\mathbf{x}_n}) \\
\mathcal{C}({\mathbf{x}_n})     &\mathcal{C}({\mathbf{x}_1}) &\mathcal{C}({\mathbf{x}_2}) &\cdots &\mathcal{C}({\mathbf{x}_{n-1}}) \\
\vdots  &\vdots &\vdots &\ddots &\vdots \\
\mathcal{C}({\mathbf{x}_2})     &\mathcal{C}({\mathbf{x}_3}) &\mathcal{C}({\mathbf{x}_4}) &\cdots &\mathcal{C}({\mathbf{x}_1}) \\
\end{array}
\right],
\end{eqnarray}
corresponding to the two-dimensional-image version of Eq.~\eqref{Eq:CirMatrX}. It can be decomposed as
\begin{eqnarray}
X=\sum_{j=1}^n V_j\otimes\mathcal{C}(\mathbf{x}_j)=\sum_{j=1}^n\sum_{k=1}^m \mathbf{x}_{jk}V_j\otimes V_k^{'}.
\end{eqnarray}
where $V_k^{'}=\sum_{l=0}^{m-1} \ket{(l-k+1)\mod m}\bra{l}$ and can also be efficiently implemented in time $O(\log (m))$.

Given the quantum oracle $O_{\mathbf{x}}$ such that $O_\mathbf{x}\ket{0}^{\otimes \lceil \log n m\rceil}=\sum_{i=1}^n\sum_{j=1}^m \sqrt{\mathbf{x}_{jk}}\ket{j}\ket{k}$ and under the assumption that $\sum_{jk}\mathbf{x}_{jk}=1$, the extended circulant Hamitonian $\tilde{X}=\ket{0}\bra{1}\otimes X+\ket{1}\bra{0}\otimes X^{\dag}$ can be efficiently simulated as shown in Theorem \ref{Theorem:ExCirSimulation} by extending the dimension $n$ to $nm$. The ability of efficiently simulating $\tilde{X}$ allows us to perform training and detection in the quantum computer as shown in the previous two subsections, except that the dimension has been extended from $n$ to $nm$.

\section{Applications}
\label{Sec:Applications}
In this section, we show how the output response state $\ket{\hat{\mathbf{y}}}$ of our quantum algorithm can be fully used to efficiently implement two important tasks related to VT: (1) object disappearance detection and (2) motion behavior matching.

\subsection{Object disappearance detection}

The task of \textit{object disappearance detection} is to detect whether the object has already disappeared or not, in the candidate image patch of the detection frame. If not, as discussed in the previous section, the set of elements in $\hat{\mathbf{y}}$ (Eq.~\eqref{Eq:DetectionPred}, namely the output in the detection phase of HCMB15 framework, will be approximately equal to that in $\mathbf{y}$ \cite{PC}. This means that in this case the elements in $\hat{\mathbf{y}}$ also approximately follow a Gaussian distribution with a peak. However, if the object has disappeared from the video frame, the distribution of the elements in $\hat{\mathbf{y}}$ becomes much more uniform \cite{PC}. This implies that in this case the quantum state $\ket{\hat{\mathbf{y}}}$ is closer to the uniform superposition state $\ket{\mathbf{1}}:=\sum_{j=1}^{n-1}\ket{j}/\sqrt{n}$. In order to clearly discriminate these two cases, we estimate the overlap between $\ket{\hat{\mathbf{y}}}$ and $\ket{\mathbf{1}}$, $P_1=\abs{\langle \hat{\mathbf{y}} \ket{\mathbf{1}}}^2$, using \textit{swap test} \cite{SSP16,QSVM14,BCWW01}, to determine and quantify the closeness between these two states. The quantum circuit of the canonical swap test for estimating the overlap of two quantum states is shown in Fig.~\ref{fig:SwapTest}, the circuit depth of which scales linearly with the qubit number of each state. The swap test has been implemented in both quantum optics \cite{GC13,Patel16,Ferreyrol13} and a gate-based quantum computer \cite{Linke17}. The task of estimating the state overlap can also be done by recent machine-learning-found quantum algorithms \cite{Cincio18}, which have shorter and even constant depth and thus may be more promising on near-term quantum computers than the canonical swap test. These algorithms experimentally exhibit more robustness and reliability, but their performances substantially depend on the connectivity and available quantum gates of the actual quantum computers.

\begin{figure}[htb]
\includegraphics[scale=0.5]{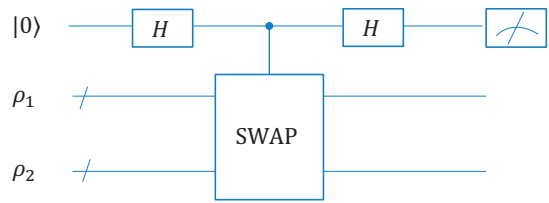}
\caption{\label{fig:SwapTest} Quantum circuit of swap test for estimating the overlap, $\textmd{\textmd{Tr}}(\rho_1\rho_2)$, between two states with density matrices $\rho_1$ and $\rho_2$, where SWAP denotes the SWAP unitary operation implemented by a series of elementary one-qubit or two-qubit gates \cite{QCQI10}. The probability of obtaining the measurement outcome $\ket{0}$ after measuring the top ancilla qubit is $\frac{1+\textmd{Tr}(\rho_1\rho_2)}{2}$, which reveals the estimate of $\textmd{Tr}(\rho_1\rho_2)$ by repeating the circuit for a sufficient number of times. In our case, we are to estimate the overlap between the two pure states $\rho_1=\ketbra{\hat{\mathbf{y}}}{\hat{\mathbf{y}}}$ and  $\rho_2=\ketbra{\mathbf{1}}{\mathbf{1}}$, $\textmd{Tr}(\rho_1\rho_2)=P_1=\abs{\langle \hat{\mathbf{y}} \ket{\mathbf{1}}}^2$.}
\end{figure}

By setting a threshold $\vartheta_1$, the object is considered as disappeared if the obtained $P_1\geq \vartheta_1$; otherwise, the object is considered still in the video frame.
Following \cite{SSP16,QSVM14,BCWW01}, $P_1$ can be obtained to accuracy $\delta$ with $O(1/\delta^2)$ repetitions of swap test, so this task only takes $O(1/\delta^2)$ copies of the response state $\ket{\hat{\mathbf{y}}}$ (as well as $\ket{\mathbf{1}}$), as opposed to $O(\sqrt{n})$ copies for sampling described in the last section. This means that, compared with the task of sampling $\ket{\hat{\mathbf{y}}}$, the task of object disappearance detection takes exponentially fewer number of $\ket{\hat{\mathbf{y}}}$, if $\delta=O(1/\polylog{n})$ is acceptable.

To show more intuitively  how different the values of $P_1$ are in the above two cases, we carried out a numerical experiment with an illustrative example, as shown in Fig.~\ref{fig:ObjectDisappearance}. In the experiment, we manually and randomly generate one training frame, and two detection frames in which the object is either still there or disappeared. In this example, $P_1=0.577$ for the case where the object exists and $P_1=0.986$ for the case where the object disappears. We then run the experiment for 50 times to yield 50 values of $P_1$ for each case. These $P_1$ values are shown in Fig.~\ref{fig:P1Comparison}. From this figure, we can see that $P_1\le 0.6$ for the case where the object exists, while $P_1 \ge 0.9$ for the case where the object disappears. Therefore, for this example, $\vartheta_1$ can satisfactorally take the value $0.75$.

\begin{figure}[htb]
\includegraphics[scale=0.4]{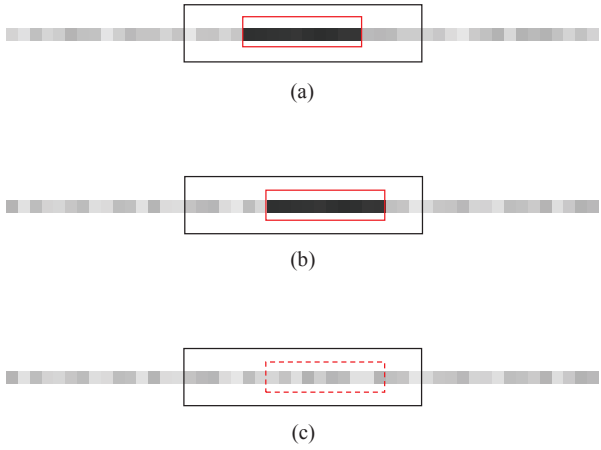}
\caption{\label{fig:ObjectDisappearance} (a) is the training frame, where the object is placed at the center. The black (bigger) rectangle denotes the image patch, and the red (smaller) rectangle denotes the object. (b) is the detection frame where the object moves 3 pixels to the right relative to that in (a) but still exists in the patch. (c) is the detection frame where the object disappears relative to that in (b), and the dashed rectangle denotes the position where the object should be. All the three frames are one-dimensional grayscale images of 50 pixels. The patch and the object are of 20 pixels and 10 pixels respectively.}
\end{figure}

\begin{figure}[htb]
\includegraphics[scale=0.58]{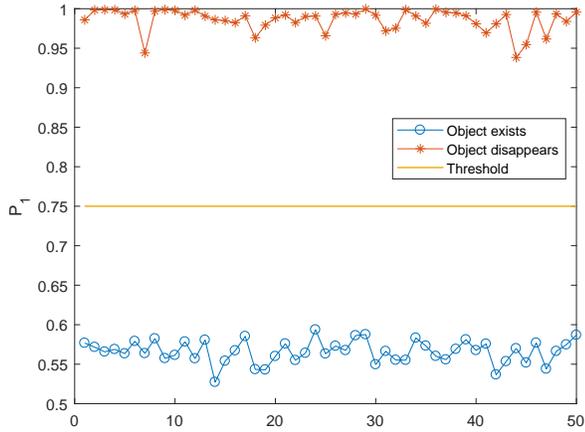}
\caption{\label{fig:P1Comparison} Comparison between the 50 values of $P_1$ for the case where the object exists and those for the case where the object disappears after running the experiment in Fig.~\ref{fig:ObjectDisappearance} for 50 times.}
\end{figure}


\subsection{Motion behavior matching}
The task of \textit{motion behavior matching} is to determine whether the motion behavior of the object in a video matches a given motion behavior template. More specifically, it is to determine whether the object in the video moves along a given path or not. To see this task visually, we present a simple example with two-dimensional images in Fig.~\ref{fig:MBM}. The procedure to implement this task using our quantum algorithm is detailed in the following steps.

\begin{figure}[htb]
\includegraphics[scale=0.37]{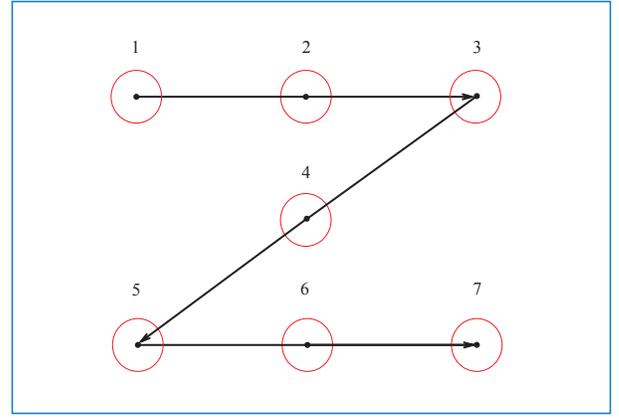}
\caption{\label{fig:MBM} A simple example for motion behavior matching. Here the rectangle with blue border denotes the frame, the red circles denote the object, the black arrowed lines correspond to the $Z$-shape motion template, and the numbers $1,2,\cdots,7$ are to mark seven selected positions in the template. The initial frame of the video contains the object located in position 1. If the object motion behavior in the video matches the template, the object in seven appropriately selected frames of the video should be located at these seven positions.}
\end{figure}

1. Train the entire initial frame of the video (instead of an image patch according to the standard HCMB15 framework) by using the training part of our quantum algorithm and get the classifier $\ket{\mathbf{w}}$.

2. Select $K$ positions in the template including the position of the object in the initial frame, generate $K$ ``template" frames by just moving the object in the initial frame to these positions, and perform the detection part of our quantum algorithm, using the $\ket{\mathbf{w}}$ obtained on these template frames, to generate $K$ response states, $\ket{\hat{\mathbf{y}}_t^1},\ket{\hat{\mathbf{y}}_t^2},\cdots,\ket{\hat{\mathbf{y}}_t^K}$.

3. Select $K$ ``actual" frames appropriately from the video, where the object is supposed to be located at the selective $K$ positions if the object motion behavior matches the template well, and perform the detection part of our quantum algorithm with the classifier $\ket{\mathbf{w}}$ on these frames, to generate $K$ response states, $\ket{\hat{\mathbf{y}}_a^1},\ket{\hat{\mathbf{y}}_a^2},\cdots,\ket{\hat{\mathbf{y}}_a^K}$. Here we assume that the information of the object moving speed is known so that these actual frames can be appropriately selected.

4. Evaluate the closeness between $\ket{\psi_t}$ and $\ket{\psi_a}$ by estimating the value of $P_2=\abs{\langle \psi_t \ket{\psi_a}}^2$ using swap test as shown in Fig.~\ref{fig:SwapTest}, where $\ket{\psi_t}=\otimes_{k=1}^K \ket{\hat{\mathbf{y}}_t^{k}}$ and $\ket{\psi_a}=\otimes_{k=1}^K \ket{\hat{\mathbf{y}}_a^{k}}$. Just as $P_1$, $P_2$ can also be estimated by shorter-depth quantum algorithms \cite{Cincio18}. Setting another threshold $\vartheta_2$, if $P_2\geq \vartheta_2$, we regard the motion behavior of the object in the video matches the template well; otherwise, the matching fails.

In step 4, $P_2$ can be obtained to accuracy $\delta$ with $O(1/\delta^2)$ copies of $\ket{\psi_t}$ and $\ket{\psi_a}$ and each copy takes $2K$ response states of same size. Therefore, $O(K/\delta^2)$ copies of response states ($\ket{\psi_t}$ and $\ket{\psi_a}$) are taken in total for implementing this task. Moreover, $O(K/\delta^2)=O(\polylog(n))$ if $K,1/\delta=O(\polylog(n))$ is acceptable. In practice, $K$ can be chosen to be very small relative to the number of frames in the video, if we just want to roughly know the object's  motion in the video. The choice for the threshold $\vartheta_2$ depends on how well we require the object's moving motion to match the template, but practically and reasonably $\vartheta_2$ should be chosen to be close to 1, e.g., $\vartheta_2=0.9$.  In addition, the accuracy $\delta$ is generally chosen to be $O(\vartheta_2)$, e.g., $\delta=\vartheta_2/10$. Therefore, the scaling $O(K/\delta^2)$ can be small in practice.

In addition to the above two applications, we expect that this algorithm can be implemented for other tasks of practical interest.

\section{Conclusions}
\label{Sec:Conculsions}

We have presented a quantum algorithm for visual tracking based on the well-known classical HCMB15 framework, which can be applied to detect the object position as done in the HCMB15 framework with quadratic speedup. Our algorithm firstly trains a quantum-state ridge regression classifier, where the optimal fitting parameters of ridge regression are encoded in the amplitudes. The classifier is then performed on the detection frame to generate a quantum state whose amplitudes encode the responses of all the candidate image patches. Taking full advantage of efficient extended circulant Hamiltonian simulation, both states can be generated in time polylogarithmic in their dimensionality, when the data matrices have low condition numbers. This demonstrates that our quantum algorithm has the potential to achieve exponential speedup over the classical counterpart. Furthermore, we have also shown how our algorithm can be applied to efficiently implement two important tasks: object disappearance detection and motion behavior matching.

We expect the techniques used in our algorithm, such as extended circulant Hamiltonian simulation, to be helpful in designing more quantum algorithms requiring manipulating circulant matrices. We also hope this algorithm may inspire more quantum algorithms for visual tracking as well as other computer vision problems.

\section*{Acknowledgements}
The authors would like to thank J. F. Henriques for helpful discussions and Bruce Hartley for proof-reading the manuscript. This work is supported by NSFC (Grant Nos. 61572081, 61672110, and 61671082). C.-H. Yu is supported by China Scholarship Council.

\appendix*
\section{Preparing the quantum state $\ket{\mathbf{y}}$ in the training phase}
\label{Appendix:preparey}
According to the definition of $\mathbf{y}$ (Eq.(\ref{Eq:TrainOutput}), for one-dimensional images, it is easy to see that $y_i=e^{-(i-1)^2/s^2}$ for $i=1,2,\cdots,\lfloor (n+1)/2 \rfloor$, and $y_i=e^{-(n+1-i)^2/s^2}$ for $i=\lfloor (n+1)/2 \rfloor+1,\lfloor (n+1)/2 \rfloor+2,\cdots, n$. Since the elements of $\mathbf{y}$ are generally not uniformly distributed, it is time consuming to create the quantum state $\ket{\mathbf{y}}=\sum_{i=1}^n y_i \ket{i}/\norm{\mathbf{y}}$ in the common way: $y_i$ are loaded into a quantum register in parallel and a controlled rotation and measurement are then conducted on an ancilla qubit so that $y_i/\norm{\mathbf{y}}$ can be written in the amplitudes. Another efficient way is referred to \cite{GR}, where $\sum_{i=i_1}^{i_2} y_i^2$ for any two $i_1$ and $i_2$, with $1\leq i_1\leq i_2 \leq n$, is required to be efficiently computable to create $\ket{\mathbf{y}}$ \cite{HHL09}. However, $y_i$ cannot satisfy this condition because there is no efficient formula for calculating $\sum_{i=i_1}^{i_2} y_i^2$ for any two $i_1$ and $i_2$ with $1\leq i_1\leq i_2 \leq n$. In the following, we present a new approach that combines both of the two ways to efficiently create $\ket{\mathbf{y}}$.

Our new approach is based on the observation that for $i=2,\cdots,\lfloor (n+1)/2 \rfloor$, $y_i$ can be approximated by integrating some Gauss function in some appropriate range, that is,
\begin{eqnarray}
y_i^2\approx \tilde{y}_i^2:=s\int_{\frac{i-2}{s}}^{\frac{i-1}{s}} e^{-2t^2} dt,
\end{eqnarray}
but $y_i^2 \leq \tilde{y}_i^2$. This is related to the error function $E(x)=\frac{2}{\sqrt{\pi}}\int_0^x e^{-t^2}dt$ ($x>0$) that can be approximated by some elementary functions. For example, $E(x)\approx G(x):=1-(a_1t+a_2t^2+a_3t^3)e^{-x^2}$ within error $2.5\times 10^{-5}$ \cite{AS}, where $t=\frac{1}{1+px}$, $p=0.47047$, $a_1=0.3480242$, $a_2=-0.0958798$, and $a_3=0.7478556$. Since this error is very small and thus negligible, we replace $E(x)$ with $G(x)$ and have
$$\tilde{y}_i^2=\frac{\sqrt{\pi}}{2\sqrt{2}}s(G(\frac{\sqrt{2}(i-1)}{s})-G(\frac{\sqrt{2}(i-2)}{s}))$$
for $i=2,\cdots,\lfloor (n+1)/2 \rfloor$. This means that
$\sum_{i=i_1}^{i_2} \tilde{y}_i^2$ for any two $2\leq i_1\leq i_2 \leq \lfloor (n+1)/2 \rfloor$ is efficiently computable. Furthermore, for $i=2,\cdots,\lfloor (n+1)/2 \rfloor$, since $y_1=1$ and $y_i=y_{n+2-i}$, we can set $\tilde{y}_1^2=1$ and $\tilde{y}_i^2=\tilde{y}_{n+2-i}^2$, thus $\sum_{i=i_1}^{i_2} \tilde{y}_i^2$ for any two $1\leq i_1\leq i_2 \leq n$ is also efficiently computable. Consequently, using the approach of \cite{GR}, we can create the state
\begin{eqnarray}
\label{Eq:stateytilde}
\ket{\tilde{\mathbf{y}}}:=\sum_{i=1}^n \frac{\tilde{y}_i}{\sqrt{\sum_{i=1}^n \tilde{y}_i^2}}\ket{i}
\end{eqnarray}
efficiently in time $O(\log n)$ \cite{GR}.

Armed with the capability of efficiently creating $\ket{\tilde{\mathbf{y}}}$, we can create the state $\ket{\mathbf{y}}$ efficiently as well, by the following procedure.

1. Add two registers to the register of $\ket{\tilde{\mathbf{y}}}$ and load the $\tilde{y}_i$ and $y_i$ in parallel to get the state
\begin{eqnarray}
\sum_{i=1}^n \frac{\tilde{y}_i}{\sqrt{\sum_{i=1}^n \tilde{y}_i^2}}\ket{i}\ket{y_i}\ket{\tilde{y}_i}.\nonumber
\end{eqnarray}
Note that $\tilde{y}_i$ and $y_i$ can be computed efficiently.

2. Add another qubit and perform the controlled rotation to have the state
\begin{eqnarray}
\sum_{i=1}^n \frac{\tilde{y}_i}{\sqrt{\sum_{i=1}^n \tilde{y}_i^2}}\ket{i}\ket{y_i}\ket{\tilde{y}_i}\left(\frac{y_i}{\tilde{y}_i}\ket{1}+\sqrt{1-\frac{y_i^2}{\tilde{y}_i^2}}\ket{0}\right).\nonumber
\end{eqnarray}

3. Undo the first step and get the state
\begin{eqnarray}
\sum_{i=1}^n \frac{\tilde{y}_i}{\sqrt{\sum_{i=1}^n \tilde{y}_i^2}}\ket{i}\left(\frac{y_i}{\tilde{y}_i}\ket{1}+\sqrt{1-\frac{y_i^2}{\tilde{y}_i^2}}\ket{0}\right).\nonumber
\end{eqnarray}

4. Measure the last qubit to ensure the outcome of $\ket{1}$ and the first register will be in the state $\ket{\mathbf{y}}$ as desired.  The success probability of this post selection is $$ \frac{\sum_{i=1}^n y_i^2}{\sum_{i=1}^n \tilde{y}_i^2}=\Theta(1),$$ because  $$\sum_{i=1}^n \tilde{y}_i^2 \approx \sum_{i=1}^n y_i^2 \mbox{ ~ and ~ } \sum_{i=1}^n y_i^2 \approx \frac{\sqrt{\pi}}{2\sqrt{2}}s$$ when $n$ is sufficiently large.  Consequently the time complexity for creating $\ket{\mathbf{y}}$  scales as $O(\log n)$.

\end{document}